\documentclass{article}
\usepackage{geometry}
\usepackage{amsmath,amssymb,amsthm} 
\usepackage[colorlinks,linkcolor=red,anchorcolor=blue,citecolor=green]{hyperref}
\usepackage{tikz}
\usetikzlibrary {arrows.meta,automata,positioning}

\usepackage{enumitem}
\setlist[enumerate]{label*=(\arabic*)}

\usepackage[capitalise]{cleveref}
\Crefname{enumi}{}{}
\crefname{enumi}{}{}

\usepackage[all]{xy}

\newtheorem{theorem}{Theorem}

\newtheorem{example}[theorem]{Example}
\newtheorem{lemma}[theorem]{Lemma}

\newtheorem{remark}[theorem]{Remark}

\newtheorem{proposition}[theorem]{Proposition}
\newtheorem{claim}[theorem]{Claim}

\DeclareMathOperator{\MCA}{\mathcal{A}}                               
\DeclareMathOperator{\MBZ}{\mathbb{Z}}                               
\DeclareMathOperator{\dupl}{dupl}                               
\DeclareMathOperator{\excl}{excl}                               
\DeclareMathOperator{\Pow}{\mathcal{P}}                               

\date{} 

\makeatletter %

\def\theenumi{\@arabic\c@enumi}
\makeatother

\begin{document}

\title{Around Don's conjecture for binary completely reachable automata}
\author{Yinfeng Zhu \thanks{Institute of Natural Sciences and Mathematics,
Ural Federal University, 620000 Ekaterinburg, Russia} \\ Yinfeng.Zhu7{@}gmail.com }

\maketitle

\begin{abstract}
	A word $w$ is called a \emph{reaching} word of a subset $S$ of states in a deterministic finite automaton (DFA) if $S$ is the image of $Q$ under the action of $w$.
	A DFA is called \emph{completely reachable} if every non-empty subset of the state set has a reaching word.
	A conjecture states that in every $n$-state completely reachable DFA, for every $k$-element subset of states, there exists a reaching word of length at most $n(n-k)$. 
	We present infinitely many completely reachable DFAs with two letters that violate this conjecture.
	A subfamily of completely reachable DFAs with two letters, is called \emph{standardized} DFAs, introduced by Casas and Volkov (2023).
	We prove that every $k$-element subset of states in an $n$-state standardized DFA has a reaching word of length $\le n(n-k) + n - 1$. 
	Finally, we confirm the conjecture for standardized DFAs with additional properties, thus generalizing a result of Casas and Volkov (2023).
\end{abstract}

\section{Introduction}

\subsection{Completely reachable automata and Don's conjecture}

An \emph{deterministic finite automaton} (DFA) is a triple $\MCA = (Q,\Sigma,\delta)$ where $Q$ and $\Sigma$ are finite non-empty sets and $\delta \colon Q \times \Sigma \to Q$ is a function. The elements of $Q$ are called \emph{states}; the elements of $\Sigma$ are called \emph{letters}; and the function $\delta$ is called the \emph{transition function} of $\MCA$.
Finite sequences over $\Sigma$ (including the empty sequence denoted by $\epsilon$) are called \emph{words}. Write $\Sigma^*$ for the set of all words over $\Sigma$. For a word $w \in \Sigma^*$, the \emph{length} of $w$ is denoted by $|w|$.

The transition function $\delta$ expands to a function $Q \times \Sigma^* \to Q$ (still denoted by $\delta$) via the recursion: for each $q \in Q$, $a \in \Sigma$, $w \in \Sigma^*$, set $\delta(q,\epsilon) = q$ and $\delta(q,wa) = \delta(\delta(q,w))$. 
The power set of $Q$ is denoted by $\Pow(Q)$. 
The transition function $\delta$ can be further expanded to a function $\Pow(Q)  \times \Sigma^* \to \Pow(Q)$ (denoted by $\delta$ again) by setting $\delta(P,w) = \{\delta(p,w) : p \in P\}$, for all subset $P \subseteq Q$ and $w \in \Sigma^*$.
The \emph{pre-image} of $P \subseteq Q$ via a word $w\in \Sigma$, denoted by $\delta^{-1}(P,w)$, is the set $\{q : \delta(q,w) \in P\}$. It does not necessarily hold that $\delta(\delta^{-1}(P,w)) = P$.

A subset $P \subseteq Q$ is called \emph{reachable} in $\MCA$ if there exists a word $w \in \Sigma^*$ such that $P = \delta(Q,w)$. For a reachable subset $P \subseteq Q$, a word $w$ is called a \emph{reaching} word of $P$, if $P = \delta(Q,w)$.
A DFA is called \emph{completely reachable} if every non-empty subset of its state set is reachable. 

The notion of complete reachability was first introduced by Bondar and Volkov  \cite{BV16}. 
This concept appears in the study of the descriptional complexity of formal languages \cite{Mas12} and is closely related to the \v{C}ern\/{y} conjecture. The famous conjecture of \v{C}ern\/{y} states that for an $n$-state DFA $(Q,\Sigma,\delta)$, if a singleton subset of $Q$ is reachable, then there exists a word $w$ of length at most $(n-1)^2$ such that $|\delta(Q,w)| = 1$. We refer to the chapter~\cite{KV21} of the ``Handbook of Automata Theory''
and the recent survey \cite{Vol22} for a summary of the state-of-the art around the \v{C}ern\'{y} conjecture. Henk Don \cite[Conjecture 18]{Don16} proposed a stronger conjecture: for an $n$-state DFA $(Q,\Sigma,\delta)$, if a $k$-element subset of $Q$ is reachable, then it can be reached with a word of length at most $n(n-k)$.
Don's conjecture was disproved in \cite[Proposition 7]{GJ19}: for every $n \ge 6$, there exist a DFA $(Q,\Sigma,\delta)$ and a subset $S \subseteq Q$ such that the length of any word that reaches $S$ is at least $\frac{2^n}{n}$. However, it is interesting to ask whether Don's conjecture hold for completely reachable DFAs. This problem was explicitly asked in \cite[Problem 4]{GJ19}.
 Ferens and Szyku{\l}a \cite[Theorem 30]{FS23} prove that for a completely reachable DFA with $n$ states, every subset with $k$ states, $0<k\le n$, can be reached in a word of length at most $2n(n-k)$. Casas and Volkov \cite[Theorem 1]{CV23} verify Don's conjecture for a subfamily of completely reachable DFAs with two letters.

\subsection{Binary completely reachable DFA and main results}

A DFA is called \emph{circular} if the action of one letter is a cyclic permutation of all states.
DFAs with two letters are called \emph{binary}. For $n \ge 3$, a binary completely reachable $n$-state DFA is circular, see \cite[Proposition 1]{CV22}.
In this article, for a circular DFA, we use $b$ to denote the letter acting as a cyclic permutation and use $a$ to denote the other letter.
We will assume that all $n$-state circular DFA have the state set $\mathbb{Z}_n$, which is the set of all residues modulo $n$, and the letter $b$ acts on $\mathbb{Z}_n$ as $\delta(q,b) = q \oplus 1$ for each $q \in \mathbb{Z}_n$,  where $\oplus$ stands for addition modulo $n$.

For a DFA $\MCA = (Q,\Sigma,\delta)$ and a word $w \in \Sigma^*$, we define the \emph{excluded set} of $w$ as
\[
	\excl(w) := \{q: |\delta^{-1}(q,w)| = 0\}
\]
and the \emph{duplicate set} of $w$ as
\[
	\dupl(w) := \{q : |\delta^{-1}(q,w)| > 1\}.
\]
For a binary completely reachable DFA $\MCA = (\mathbb{Z}_n, \{a,b\},\delta)$, one can easily obtain that $|\excl(a)| = |\dupl(a)| = 1$. We may and will assume that $\excl(a) = \{0\}$. Futher, $(\mathbb{Z}_n, \{a,b\},\delta)$ is called \emph{standardized} if $\dupl(a) = \{\delta(0,a)\}$.

For a standardized completely reachable DFA, we establish an upper bound for the length of shortest words that reach a non-empty subset of the state set.
\begin{theorem} \label{thm:upperBound}
	Let $\MCA = (\MBZ_n, \{a,b\}, \delta)$ be a standardized completely reachable DFA. 
	Every non-empty subset $S \subseteq \MBZ_n$ is reachable with a word of length at most $n(n-|S|) + n - 1$.
\end{theorem}

Assume that $\MCA = (\mathbb{Z}_n, \{a,b\}, \delta)$ is a standardized completely reachable DFA. Denote the set $\{\delta(0,a^i) : i \ge 1\}$ by $O(\MCA)$, is called the \emph{orbit} of $\MCA$. 
The subgroup of $(\mathbb{Z}_n, \oplus)$ generated by $O(\MCA)$, denoted by $K(\MCA)$ or simply $K$, is called the \emph{orbit subgroup} of $\MCA$. 
We obtain a generalization of \cite[Theorem 1]{CV23}.

\begin{theorem} \label{thm:Don}
	Let $\MCA = (\MBZ_n, \{a,b\}, \delta)$ be a standardized completely reachable DFA. 
	If $|K(\MCA)| \ge \frac{n}{2}$, then $\MCA$ fulfills Don's conjecture.
\end{theorem}

To the best of our knowledge, it is still unknown whether Don's conjecture holds for standardized completely reachable  DFAs or not. But we find infinite many binary completely reachable DFAs do NOT fulfill Don's conjecture and thus we answer \cite[Problem 4]{GJ19} negatively.
\begin{theorem} \label{thm:non-Don}
	There are infinitely many binary completely reachable DFAs that do not fulfill Don's conjecture.
\end{theorem}

The remaining of this note will proceed as follows.
In \cref{sec:UpperBound}, we prove \cref{thm:upperBound}. In \cref{sec:fulfills}, we present a proof of \cref{thm:Don}. 
In \cref{sec:doNotfulfill}, we describe experiments performed to gain some binary automata that do not fulfill Don's conjecture. Based on these experiments, we construct an infinite series of binary automata that do not fulfill Don's conjecture. At the end, we summarize our results and discuss the relation between them in \cref{sec:conclusion}.

\section{An upper bound of the length of shortest reaching words in standardized completely reachable automata} \label{sec:UpperBound}

For a group $G$ and a subset $S \subseteq G$, we write $\langle S \rangle$ for the \emph{subgroup generated by $S$}, which is the smallest subgroup of $G$ containing $S$.
For a subset $X \subseteq \MBZ_n$ and $a \in \MBZ_n$, write $X \oplus a$ for $\{x \oplus a : x \in X\}$.

Let $\MCA = (\MBZ_n,\{a,b\},\delta)$ be a standardized completely reachable DFA. 
Let $H_0$ be the trivial subgorup of $(\mathbb{Z}_n, \oplus)$. For every $i \ge 1$, define the subgroup $H_i (= H_i(\MCA))$ of $(\mathbb{Z}_n, \oplus)$ as
\[
	H_i := \langle U_i \cup H_i \rangle
\]
where $U_i := \delta(H_i,a) \setminus H_i$.
Since $\MCA$ is completely reachable, by \cite[Proposition 1]{CV22}, there exists a positive integer $\ell = \ell(\MCA)$ such that
\[
	H_0 \subsetneq H_1 \subsetneq \cdots \subsetneq H_{\ell} = (\mathbb{Z}_n, \oplus).
\]
For a non-empty proper subset $S$ of $\MBZ_n$, we define 
\begin{itemize}
	\item $m(S) := \min \{i : H_i \cap S \notin \{\emptyset, H_i\}\}$;
	\item $t(S) := \max \{i : \text{$H_{m(S)} \cap S$ is a union of $H_i$-cosets}\}$.
\end{itemize}
Since $H_{m(S)} \cap S \neq H_{m(S)}$, it holds that
\begin{equation} \label{eq:tm}
	0 \le t(S) < m(S) \le \ell(\MCA).
\end{equation}

For two subsets $P,P' \subseteq \MBZ_n$, if there exists a word $w \in \Sigma^*$ such that $\delta(P,w) = P'$, then we say that $P$ is a \emph{$w$-predecessor} of $P'$, or simply a \emph{predecessor} of $P'$.

\begin{lemma} \label{lem:A}
	Let $\MCA = (\MBZ_n,\{a,b\},\delta)$ be a standardized completely reachable DFA.
	Let $S$ be a non-empty proper subset of $\MBZ_n$. 
	Let $t = t(S)$ and $m = m(S)$.
	There exists an $H_t$-coset $C$ and $u \in U_t$ such that $C \subseteq H_m \setminus S$ and $C \oplus u \subseteq H_m \cap S$.
\end{lemma}
\begin{proof}
	By the definition of $m$ and $t$, there exist some $H_t$-cosets in $H_m \setminus S$. 

	Assume, for a contradiction, that for every $u \in U_t$ and $H_t$-coset $C \subseteq H_m \setminus S$, we have $C \oplus u$ is not a subset of $H_m \cap S$.
	Since both $H_m$ and $H_m \cap S$ are unions of $H_t$-cosets, it holds that $H_m \setminus S$ is also a union of $H_t$-cosets.
	For every $u \in U_t$, the set $C \oplus u$ is an $H_t$-coset and then is a subset of $H_m \setminus S$. Since $H_{t+1}$ is generated by $H_t \cup U_t$, the set $H_m \setminus S$ is a union of $H_{t+1}$-cosets. This implies $H_m \cap S$ is a union of $H_{t+1}$-cosets which contradicts with the maximality of $t$.
\end{proof}

\begin{lemma} \label{lem:B}
	Let $H$ and $H'$ be two subgroup of $(\MBZ_n, \oplus)$ such that $H \subseteq H' \subseteq \MBZ_n$.
	Let $C$ be an $H$-coset such that $C \subseteq H'$. 
	Then there exists an integer $i$ such that $H \oplus i = C$ and $0 \le i \le \frac{n}{|H|} - \frac{n}{|H'|}$.
\end{lemma}
\begin{proof}
	Since $C$ is an $H$-coset, there exists a non-negative integer $i \le \frac{n}{|H|}-1$ such that $H \oplus i =C$. Meanwhile, $C \subseteq H'$ and $C \oplus i = H \subseteq H'$ implies that $i$ is a multiple of $\frac{n}{|H'|}$. Then $0 \le i \le \frac{n}{|H|} - \frac{n}{|H'|}$.
\end{proof}


\begin{proposition} \label{prop:1}
	Let $\MCA = (\MBZ_n,\{a,b\},\delta)$ be a standardized completely reachable DFA.
	Let $S$ be a non-empty proper subset of $\MBZ_n$.
	Then there exists a word $w$ of length at most $\frac{n}{|H_{t(S)}|} - \frac{n}{|H_{m(S)}|} + 1$ such that there exists a $w$-predecessor $R$ of $S$ and one of the following conditions holds
	\begin{enumerate}
		\item $|R| = |S| + 1$;
		\item $m(R) \le t(S) < m(S)$.
	\end{enumerate}
\end{proposition}
\begin{proof}
	Let $t = t(S)$ and $m = m(S)$. Using \cref{lem:A}, 
	we can take an $H_t$-coset $C$ and $u \in U_t$ such that $C \subseteq H_m \setminus S$ and $C\oplus u  \subseteq H_m \cap S$.
	By \cref{lem:B}, let $C = H_t \oplus i$ such that $0 \le i \le \frac{n}{|H_t|} - \frac{n}{|H_m|}$.
	Let $w = ab^i$ and $R = \delta^{-1}(S, w)$. Note that the length of $w$ is at most $ \frac{n}{|H_t|} - \frac{n}{|H_m|} +1$.
	Since $C \cap S = \emptyset$, we have $0 \notin \delta^{-1}(S, b^i)$ and then the set $R$ is a predecessor of $S$.

	Now we need to check that $R$ satisfies the required condition. If $|R| > |S|$, we are done. Otherwise, we have $\delta(0,a) \notin \delta(R,a) = \delta^{-1}(S, b^i)$ which implies $0 \notin R$ and then 
	\begin{equation} \label{eq:properSubset}
		R \cap H_t \neq H_t.
	\end{equation}
	Since $C \oplus u \subseteq S$, we have $H_t \oplus u \subseteq \delta^{-1}(S, b^i) = \delta(R,a)$. Then $\delta^{-1}(u,a) \in R \cap H_t$ which implies
	\begin{equation} \label{eq:nonemptyset}
		R \cap H_t \neq \emptyset.
	\end{equation}
	Combining \cref{eq:tm,eq:properSubset,eq:nonemptyset}, $m(R) \le t(S) < m(S)$.
\end{proof}

Now we are ready to prove \cref{thm:upperBound}.

\begin{proof}[Proof of \cref{thm:upperBound}]
	Since $\MBZ_n$ is reachable by the empty word, we can assume that $S \neq \MBZ_n$.
	Write $s$ for the size of $S$.

	We will construct a sequence of subset $S_0, S_1, \ldots, S_{n-s}$ of $\MBZ_n$ and a sequence of words $w_1, \ldots w_{n-s}$ such that 
	\begin{itemize}
		\item $S_0 = S$, $S_{n-s} = \MBZ_n$;
		\item $|S_i| = s + i$, for $0 \le i \le n-s$;
		\item $\delta(S_i, w_i) = S_{i-1}$, for $1 \le i \le n-s$.
	\end{itemize}

	Let $i$ be an integer such that $1 \le i < n-s$. 
	Assume that $S_j$ and $w_j$ are already defined for $0 \le j \le i$. Now we will define $S_{i+1}$ and $w_{i+1}$ as follows. 

	Let $T_{i,0} = S_i$.
	For $j > 1$, if $T_{i,j-1}$ has been defined and $|T_{i,j-1}| = |T_{i,0}| = s+i$, then, applying \cref{prop:1} on $T_{i,j-1}$, there exists a subset $R \subseteq \MBZ_n$ and a word $w'$ of length at most $\frac{n}{|H_{t(T_{i,j-1})}|} - \frac{n}{|H_{m(T_{i,j-1})}|} + 1$ such that $\delta(R,w') =T_{i,j-1}$ and one of the following conditions holds
	\begin{itemize}
		\item $|R| > |T_{i,j-1}|$,
		\item $m(R) \le t(T_{i,j-1}) < m(T_{i,j-1})$.
	\end{itemize}
	Define $T_{i,j} = R$ and $w_{i,j} = w'$.

	Write $k_i$ for the maximum integer $j$ such that $T_{i,j}$ has been defined.
	Note that $k_i$ is well-defined, otherwise 
	\[
		m(T_{i,0}), m(T_{i,1}), \ldots
	\]
	is an infinite strictly decreasing non-negative integer sequence which is nonexistent.
	Observe that $|T_{i,k_i}| = |S_i| + 1$. Define $S_{i+1} = T_{i,k_i}$ and $w_i = w_{i,k_i} \cdots w_{i,1}$.

	Let $w = w_{n-s} \cdots w_1$. It is clear that $\delta(\MBZ_n,  w)= S$. To complete the proof, we estimate the length of $w$ as follows:
	\begin{align*}
		|w| &= \sum_{i = 1}^{n-s} \sum_{j = 1}^{k_i} |w_{i,j}| \\
		&\le \sum_{i = 1}^{n-s} \sum_{j = 1}^{k_i} \frac{n}{|H_{t(T_{i,j-1})}|} - \frac{n}{|H_{m(T_{i,j-1})}|} + 1\\
		& \le \sum_{i = 1}^{n-s} \left( \frac{n}{|H_{t(T_{i,k_i - 1})}|} - \frac{n}{|H_{m(T_{i,0})}|} + \sum_{j = 1}^{k_i} 1 \right) \quad\quad\quad\quad \text{[\cref{prop:1}, $m(T_{i,j+1}) \le t(T_{i,j})$]} \\
		& \le \sum_{i = 1}^{n-s} \left( n - 1 + \sum_{j = 1}^{k_i} 1 \right) \\
		& = (n-1)(n-s) + \sum_{i = 1}^{n-s} \sum_{j = 1}^{k_i} 1 \\
		& = (n-1)(n-s) + \sum_{i = 1}^{n-s} \left| \{H_t : t = t(T_{i,j-1}), 1 \le j \le k_i-1\} \right|\\
		& \le (n-1)(n-s) + \sum_{t = 0}^{\ell(\MCA)} \frac{n}{|H_{t}|} \\
		& \le (n-1)(n-s) + 2n -1 = n(n-s) + n -1.
	\end{align*}
\end{proof}

\section{DFAs that fulfills Don's conjecture}\label{sec:fulfills}

Let $\MCA= (\MBZ_n, \{a,b\} , \delta)$ be a DFA.
We say that a word $w \in \{a,b\}^*$ \emph{expands} a proper non-empty subset $S \subseteq \MBZ_n$ if there exists a subset $R \subseteq \MBZ_n$ such that $|R| > |S|$ and $\delta(R,w) = S$ and we say $S$ is \emph{$|w|$-expandable}.
Recall that the orbit $O(\MCA)$ of $\MCA$ is the set $\{\delta(0,a^i) : i \ge 1\}$. 
The orbit subgroup $K = K(\MCA)$ of $\MCA$ is the subgroup of $(\mathbb{Z}_n, \oplus)$ generated by $O(\MCA)$. 

The following result is established in \cite[Proposition 7]{CV23}.
\begin{proposition} \label{coro:CV23+}
	Let $\MCA = (\MBZ_n, \{a,b\},\delta)$ be a standardized completely reachable DFA. Every non-empty subset $S$ of $\MBZ_n$ which is not a union of $K$-cosets is $n$-expandable.
\end{proposition}

Moreover, the proof of the above proposition in \cite{CV23}  is able to obtain a stronger result as follows.

\begin{proposition} \label{prop:expand}
	Let $\MCA = (\MBZ_n, \{a,b\},\delta)$ be a standardized completely reachable DFA and $K$ its orbit subgroup. Let $S$ be a non-empty subset of $\MBZ_n$. If there exists a $K$-coset $C$ such that $S \cap C \notin \{\emptyset, C\}$, then there exists a word $w = a^sb^t$ of length $\le n$ that expands $S \cap C$. 
\end{proposition}

If one drops the condition that $S$ is not a union of $K$-cosets in \cref{coro:CV23+}, then $S$ is not necessarily $n$-expandable, see \cite[Example 2]{CV23}.
However, with a little more effort, we can obtain the following lemma.

\begin{proposition} \label{prop:expand+1}
	Let $\MCA = (\MBZ_n, \{a,b\},\delta)$ be a standardized completely reachable DFA and $K$ its orbit subgroup. If $K = (2\MBZ_n, \oplus)$, then  
	\begin{enumerate}
		\item for a subset $S \subseteq \MBZ_n$ which is not a union of $K$-cosets, $S$ is $n$-expandable;
		\item $K \oplus 1$ is $n$-expandable. Moreover there exists a word $a^sb^{t-1}a$ with $s+t \le n$ that expands $K \oplus 1$;
		\item $K$ is $(n+1)$-expandable.
	\end{enumerate}
\end{proposition}
\begin{proof}

	\begin{enumerate}
		\item 	If $S$ is not a union of $K$-cosets, then there exists a $K$-coset $C$ such that $S \cap C \notin \{\emptyset, C\}$. Using \cref{prop:expand}, there exists a word $w$ of length at most $n$ that expands $S \cap C$. It is clear that $w$ also expands $S$. Hence $S$ is $n$-expandable.
		\item 	Let $T = \delta^{-1}(K \oplus 1, a)$. Since $0 \notin  K \oplus 1$, the subset $T$ is a predecessor of $ K \oplus 1$. It is clear that $\delta(T,b)$ is not a union of $K$-cosets. 
		Note that $(K \oplus 1) \cap \delta(T,b) \notin \{\emptyset, K \oplus 1\}$.
		Applying \cref{prop:expand} for $S = \delta(T,b)$ and $C = K \oplus 1$, there exists a word $w = a^sb^t$ that expands $(K \oplus 1) \cap \delta(T,b)$ and also expands $\delta(T,b)$.
		Note that $\delta(0,a) \notin K \oplus 1$.
		Since $w$ expands $(K \oplus 1) \cap \delta(T,b)$, we have $t \ge 1$. Hence, the word $a^sb^{t-1}a$ expands $K \oplus 1$.
		\item 	Since $\delta(K \oplus 1, b) = K$, $K$ is $(n+1)$-expandable.
	\end{enumerate}

\end{proof}

Let $\MCA = (\MBZ_n, \{a,b\}, \delta)$ be a standardized completely reachable automaton such that its orbit group $K = (2\MBZ_n, \oplus)$. Let $k$ be an integer such that $1 \le k \le n$.
Using \cref{prop:expand+1} repeatedly, one can deduce that every non-empty $k$-element subset of $\MBZ_n$ is reachable with a word of length at most $n(n-k) + 1$. In the rest of this section, we will improve this upper bound to $n(n-k)$.

\begin{lemma} \label{lem:lemon}
	Let $p,q$ be two distinct states in $\MBZ_n$.
	If $K = (2\MBZ_n, \oplus)$ and $|O(\MCA)| > 1$, then $\MBZ_n \setminus \{p,q\}$ is $(n-1)$-expandable.
\end{lemma}

\begin{proof}
	Without loss of generality, we assume that $p < q$.

	\medskip
	\noindent{\textbf{Case 1.} $p \le n-3$}

	Let $T = \delta^{-1}(\MBZ_n \setminus \{p,q\}, b^p)$. Note that $0 \notin T$ and $\{\delta(0,a), \delta(0,a^2)\} \cap T \neq \emptyset$. Then either $ab^p$ 
	or $a^2b^p$ expands $Q \setminus \{p,q\}$. Hence $\MBZ_n \setminus \{p,q\}$ is $(n-1)$-expandable, since $2+p \le n-1$.

	\medskip
	\noindent{\textbf{Case 2.} $p = n-2$}

	Note that $q = n-1$ and $\delta(0,a) \neq 1$. We have $\MBZ_n \setminus \{p,q\} = \delta(\{2,3,\ldots, n-1\},b^{n-2})$ and then $ab^{n-2}$ expands $\MBZ_n \setminus \{p,q\}$. Hence $\MBZ_n \setminus \{p,q\}$ is $(n-1)$-expandable.
\end{proof}

For a positive integer $n$ and an integer $a$, we denote the residue of $a$ modulo $n$ by $\overline{a}$.

\begin{lemma} \label{lem:extreme}
	Assume that $O(\MCA) = \{d\}$ and $K = (2\MBZ_n, \oplus)$.  Let $S \subseteq \MBZ_n$ such that $|S| > \frac{n}{2}$. If $S$ is not $(n-1)$-expandable, then $S = \MBZ_n \setminus \{\overline{n-1 - jd} : 0 \le j \le n - |S| - 1\}$.  
\end{lemma}
\begin{proof}

	Since $|S| = |\MBZ_n \setminus \{\overline{n-1 - jd} : 0 \le j \le n - |S| - 1\}|$, it is sufficient to prove $\MBZ_n \setminus S \subseteq \{\overline{n-1 - jd} : 0 \le j \le n - |S| - 1\}$.

	Take an arbitrary state $x \notin S$. Assume $x \in K \oplus i$, where $i \in \{0,1\}$. Since $O(\MCA) = \{d\}$, we have $K \oplus i = \{\overline{x + td} : 0 \le t \le \frac{n}{2}\}$. Since $|S| > \frac{n}{2}$, $S \cap (K \oplus i) \neq \emptyset$.
	Let $j$ be the least non-negative integer such that $\overline{x + (j+1)d} \in S$. It is clear that $j \le n - |S| -1$.
	Let $T = \delta^{-1}(S, b^{\overline{x + jd}})$. Note that $0 \notin T$ and $d \in T$. And then $ab^{\overline{x+jd}}$ expands $S$. Since $S$ is not $(n-1)$-expandable, this implies $\overline{x + jd} = n - 1$ which is equivalent to $x \in \{\overline{n-1 - jd} : 0 \le j \le n - |S| - 1\}$.
\end{proof}

\begin{lemma}\label{lem:acnes}
	Assume that $O(\MCA) = \{d\}$, $K = (2\MBZ_n, \oplus)$ and $n \ge 10$. There exists a word $w$ that reaches $K$ such that $|w| \le \frac{1}{2}n^2$.  
\end{lemma}
\begin{proof}
	Let $w$ be a shortest word that reaches $K$. Since $K \nsubseteq \delta(\MBZ_n,a)$, we have $w = w'b$ and $\delta (\MBZ_n, w') = K \oplus 1$. It is sufficient to prove $|w| \le \frac{1}{2} n^2$.

	Assume, for a contradiction, that $|w| \ge \frac{1}{2} n^2 + 1$.  By applying \cref{prop:expand+1}, we have 
	\[
			\xymatrix{
				K \oplus 1 & T \ar[l]_{a}  & S_{\frac{n}{2}-1} \ar[l]_{a^mb^{t-1}}  & \cdots \ar[l]_{a^{m_{\frac{n}{2}-1}}b^{t_{\frac{n}{2}-1}}} & S_{1} \ar[l]_{a^{m_2}b^{t_2}} & \MBZ_n \ar[l]_{a^{m_1}b^{t_1}} \\
			}
	\]
	such that $|S_i| = n - i$ for all $i \in \{1, 2, \ldots, \frac{n}{2}-1\}$ and
	\begin{align} 
		m + t &\le n \label{eq:upper1}\\
		m_i+t_i &\le n \label{eq:upper2}
	\end{align}
	for all $i \in \{1,2,\ldots, \frac{n}{2}-1\}$. 
	Recall that $w$ is a shortest word such that $\delta(\MBZ_n, w) = K$. Then
	\begin{equation} \label{eq:word}
		|a^{m_1}b^{t_1} a^{m_2}b^{t_2} \cdots a^{\frac{n}{2}-1}b^{t_{\frac{n}{2} - 1}}a^mb^{t-1}a| = (m+t) + \sum_{i=1}^{\frac{n}{2}-1} (m_i + t_i)  \ge |w| \ge n\frac{n}{2} + 1
	\end{equation}
	Due to $O(\MCA) = \{d\}$, we have $m = 1$ and $m_i = 1$ for all $i \in \{1, 2, \ldots, \frac{n}{2}-1\}$.
	Then inequalities \cref{eq:upper1,eq:upper2,eq:word} only hold when $t = n -2$ and $t_i = n-1$ for all  $i \in \{1,2,\ldots, \frac{n}{2}-1\}$. 

	Observe that $K$ is not $n$-expandable and $S_i$ is not $(n-1)$-expandable for $i \in \{1, \ldots, \frac{n}{2}-1\}$ (Otherwise, by \cref{prop:expand}, we can construct a word $u$ of length at most $\frac{1}{2}n^2-1$ such that $\delta(\MBZ_n,u) = K \oplus 1$). 
	By \cref{lem:extreme}, $S_i = \MBZ_n \setminus \{\overline{n-1 - jd} : 0 \le j \le i-1\}$ for all $i \in \{1,2,\ldots, \frac{n}{2}-1\}$.

	Inducting on $i$, we will establish the following:
	\begin{claim}\label{claim: arc}
		For each $i \in \{1,2, \ldots, \frac{n}{2}-2\}$, $\delta(\overline{n-1-(i-1)d},a) = \overline{-id}$.
	\end{claim}
	\begin{proof}[Proof of \cref{claim: arc}]
		Define $T_i$ to be the set $\delta(S_i,a)$, for all $i \in \{1,2,\ldots, \frac{n}{2}-1\}$. Since $\delta(T_i,b^{n-1}) = S_{i+1}$, we have $T_i = \MBZ_n \setminus \{0, \overline{-d}, \ldots, \overline{-id}\}$. 
		Write $Q' = \MBZ_n \setminus \{0\}$ and $S'_i = S_i \setminus \{0\}$  for all $i \in \{1,2,\ldots, \frac{n}{2}-1\}$.
		Since $a$ acts on $Q'$ as a permutation, the letter $a$ sends $S'_i$ to $T_i$ as a bijection for all $i \in \{1,2,\ldots, \frac{n}{2}-1\}$.
		And then $a$ sends $Q' \setminus S_i$ to $Q' \setminus T_i$ as a bijection. 

		For $i = 1$, consider that $a$ sends $Q' \setminus S_1  = \{n-1\} $ to $Q' \setminus T_1 = \{n-d\}$ as a bijection.
		That is $\delta(\overline{n-1}, a) = \overline{-d}$.

		For $i \ge 2$, consider that $a$ sends $Q' \setminus S_i = \{\overline{n-1 - (j-1)d)} : 1 \le j \le i\}$ to $Q' \setminus T_i = \{\overline{-jd} : 1 \le j \le i\}$ as a bijection. By induction hypothesis, we have $\delta(\overline{n-1 - (j-1)d}, a) = \overline{-jd}$ for each $j \in \{1, \ldots, i -1\}$. Then $\delta(\overline{n-1 - (i-1)d},a) = \overline{-id}$.
	\end{proof}

	By \cref{claim: arc}, $\overline{n-1 - (i-1)d} \notin T$ for all $j \in \{1, \ldots, \frac{n}{2}-2\}$. Observe that $0,d \notin T$. Then $T = \{\overline{d-1}, \overline{2d-1}\} \cup \{\overline{jd} : 2 \le j \le \frac{n}{2} - 1\}$. 
	Observe that $ab^d$ expands $T$ and then $ab^dab$ expands $K$. Since $K$ is not $n$-expandable, it holds that $d \ge n-2$.
	Recall that $K = (2\MBZ_n, \oplus)$ and $\mathcal{O}(\MCA) = \{d\}$. Hence $d$ is an even number and $d = n-2 = t - 1$.
	Let $T'$ be the set such that $\delta(T',b^{t-1}) = T$. One can calculate that $T' = \{n-1, d-1\} \cup\{\overline{jd} : 1 \le j \le \frac{n}{2} -2\}$. 
	Since $S_{\frac{n}{2}-1}$ has only one odd number $d-1$ and $T'$ has at least two even numbers besides $d$, by \cref{claim: arc}, this is a contradiction for $\delta(S_{\frac{n}{2}-1},a) = T'$. 
\end{proof}

\begin{proof}[Proof of \cref{thm:Don}]
	In the case that $|O(\MCA)| > 1$, the claim of the theorem follows from \cref{prop:expand+1,lem:lemon}.
	In the case that $|O(\MCA)| = 1$ and $n \ge 10$, it follows from \cref{prop:expand+1,lem:acnes}.
	In the case that $|O(\MCA)| = 1$ and $n \le 8$, it is proved by enumeration, see \cref{table:1} below. 
\end{proof}

\section{DFAs do not fulfill Don's conjecture} \label{sec:doNotfulfill}

Using naive breadth-first search, we have verified Don's conjecture for binary completely reachable automata and standardized completely reachable automata with at most $10$ states. The number of non-isomorphic DFAs is presented in \cref{table:1}.
Let $\mathcal{B}_8$, shown in \cref{fig:8state}, be one of $8$-state binary completely reachable DFAs that do not fulfill Don's conjecture. 

\begin{figure}
\centering
\begin{tikzpicture}[->,>={Stealth[round]},shorten >=1pt,
                    auto,node distance=2cm,on grid,semithick,
                    inner sep=2pt,bend angle=45]
  \node[state] (A)              {$0$};
  \node[state] (B) [right=of A] {$1$};
  \node[state] (C) [right=of B] {$3$};
  \node[state] (D) [right=of C] {$7$};
  \node[state] (E) [right=of D] {$4$};
  \node[state] (F) [right=of E] {$6$};
  \node[state] (G) [below=of F] {$2$};
  \node[state] (H) [left=of G]  {$5$};

  \path [every node/.style={font=\footnotesize}]
        (A) edge (B)
        (B) edge (C)
        (C) edge (D)
        (D) edge (E)
        (E) edge (F)
        (F) edge (G)
        (G) edge (H)
        (H) edge (E);
\end{tikzpicture}
\caption{The action of $a$ of the $8$-state DFA $\mathcal{B}_8$. The shortest word that reaches $\{1,2,3,5,6,7\}$ is the word $a^2b^5ab^5ab^3$.}
\label{fig:8state}
\end{figure}
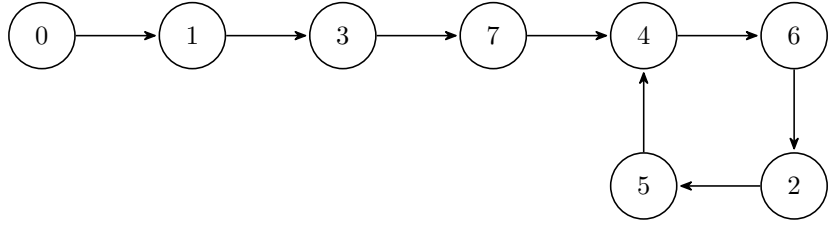

\begin{table}[htbp]
\caption{The number of non-isomorphic automata that do not fulfull Don's conjecture.}
\label{table:1}
\begin{center}
\begin{tabular}{lrrrr}
      \hline\hline
      the number of states & $\le 7$ & $8$ & $9$ & $10$ \\\hline
      standardized completely reachable  & $0$ & $0$ & $0$ & $0$\\
      binary completely reachable & $0$ & $68$ & $0$ & $9210$ \\\hline
\end{tabular}
\end{center}
\end{table}

\begin{example} \label{ex:counterexample}
Let $n \ge 10$ be an even positive integer.
Let $\MCA_n = (\mathbb{Z}_n,\{a,b\},\delta)$ be the circular automaton such that 
\begin{itemize}
	\item $\delta(q,b) = q \oplus 1$, for every $q \in \mathbb{Z}_n$;
	\item $\delta(0,a) = n-3, \delta(n-3,a) = 1, \delta(1,a) = n-2, \delta(n-2,a) = n/2, \delta(n/2,a) = n-4, \delta(n-4,a) = n-1, \delta(n-1,a) = n/2$;
	\item $\delta(q,a) = q$, for every $q \in \mathbb{Z}_n \setminus \{0, 1, n/2, n-1, n-2, n-3, n-4\}$.
\end{itemize}
The action of $a$ is shown in \cref{fig:non-Don}.
In the rest of this section, we will show that the automaton $\MCA_n$ is completely reachable (\cref{lem:CR}) and does not fulfill Don's conjecture (\cref{lem:SRW}).
\end{example}

\begin{figure}
\centering
\begin{tikzpicture}[->,>={Stealth[round]},shorten >=1pt,
                    auto,node distance=2cm,on grid,semithick,
                    inner sep=2pt,bend angle=45]
  \node[state] (A)              {$0$};
  \node[state] (B) [right=of A] {$n-3$};
  \node[state] (C) [right=of B] {$1$};
  \node[state] (D) [right=of C] {$n-2$};
  \node[state] (E) [right=of D] {$n/2$};
  \node[state] (F) [right=of E] {$n-4$};
  \node[state] (G) [right=of F] {$n-1$};

  \path [every node/.style={font=\footnotesize}]
        (A) edge              (B)
        (B) edge              (C)
        (C) edge              (D)
        (D) edge              (E)
        (E) edge              (F)
        (F) edge              (G)
        (G) edge [bend left]  (E);
\end{tikzpicture}
\caption{The action of $a$ in the DFA $\MCA_n$. The states that are fixed by $a$ are omitted.}
\label{fig:non-Don}
\end{figure}
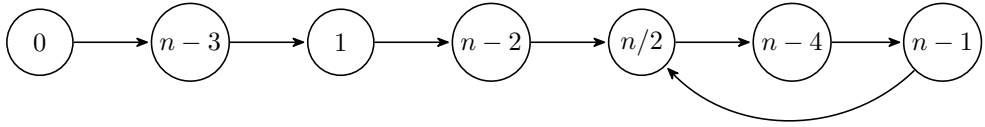

For an automaton $(Q,\Sigma,\delta)$, a non-empty subset $W$ of $Q$ is called a \emph{witness} if it is unreachable and has the maximal size of all unreachable subsets of $\MBZ_n$. This concept is introduced by Ferens and Szykuła \cite[Defnition 4]{FS23}. The following lemma is based on an important observation on witnesses \cite[Corollary 6]{FS23}: the union of two distinct witnesses is the whole state set.

\begin{lemma} \label{lem:CR}
	For each positive even integer $n$, the automaton $\MCA_n$ is completely reachable.
\end{lemma}
\begin{proof}
	Assume $\MCA_n$ is not completely reachable. Take $S$ to be a witness of $\MCA$.
	Observe that $s \in S$ if and only if $s \oplus \frac{n}{2} \in S$, otherwise, $b^sa$ expands $S$. 
	Since $\delta(S,a^i)$ is also a witness of $\MCA$ for every nonnegative integer $i$, we can assume that $0,n/2 \notin S$.
	Let $T$ and $T'$ be the subsets of $\mathbb{Z}_n$ such that $\delta(T,b) = S$ and $\delta(T',ab)=S$. 
	Note that $T$ and $T'$ are witnesses and $n-1 \notin T \cup T'$ which contradicts \cite[Corollary 6]{FS23}.
\end{proof}

Let $\Pi$ be the set of words $\{ab^i : i \ge 0 \}$. For $S,S' \subseteq \mathbb{Z}_n$, the subset $S$ is called a \emph{$\Pi$-predecessor} of $S'$ if $S$ is a $w$-predecessor of $S'$ for some word $w \in \Pi$. The following lemma contains some elementary facts of $\Pi$-predecessors in $\MCA_n$.

\begin{lemma} \label{lem:predecessor}
	In $\MCA_n$, the following statements hold.
	\begin{enumerate}
		\item \label{itm:1} $\mathbb{Z}_n \setminus \{n-1,n-2\}$ is the unique $\Pi$-predecessor of $\mathbb{Z}_n \setminus \{\frac{n}{2}-1, n-1\}$.
		\item \label{itm:2} If $\mathbb{Z}_n \setminus X$ is  a $\Pi$-predecessors of $\mathbb{Z}_n \setminus \{n-1,n-2\}$, then the minimum element of $X$ is either $n-3$ or $n-4$
		\item \label{itm:3} Let $x$ be the minimum element in a non-emptyset propoer subset $X \subseteq \MBZ_n$. The length of a word that reaches $\MBZ_n / X$ is at least $x+1$. 
	\end{enumerate} 
\end{lemma}
\begin{proof}
	Let $S = \mathbb{Z}_n \setminus \{\frac{n}{2}-1, n-1\}$ and $T = \mathbb{Z}_n \setminus \{n-1,n-2\}$.
	\begin{enumerate}
		\item Let $R$ be a $\Pi$-predecessor of $S$. Let $ab^i$ be a word such that $\delta(R, ab^i)=S$. Then $\delta(R,a) = \delta^{-1}(S,b^i)$ do not contain states $0$ and $n/2$. Hence, $R = \mathbb{Z}_n \setminus \{n-1, n-2\}$ as wanted.
		\item Let $R$ be a $\Pi$-predecessor of $T$. Let $ab^i$ be a word such that $\delta(R, ab^i)=S$. Let $P = \delta(R,a) = \delta^{-1}(S,b^i)$. Either $\MBZ_n \setminus P = \{0,1\}$ or $\MBZ_n \setminus \{0,n-1\}$.	Since $\delta(\{n-1,n-2\},a) = n/2$, the set $\MBZ_n \setminus R$ equals 
		one of the six subsets: $\{n-3\}$, $\{n-3, n-1\}$, $\{n-3, n-2\}$, $\{n-4\}$, $\{n-4,n-1\}$ and $\{n-4, n-2\}$.
		\item Let $R$ be an $ab^i$-predecessor of $\MBZ_n \setminus X$ for some integer $i$. It is clear that $0 \notin \delta (R,a) = \delta^{-1}(\MBZ_n \setminus X, b^i)$. Then $i \ge x$. Since $X \neq \emptyset$, every word that reaches $\MBZ_n \setminus X$ contains the letter $a$.
		Hence, the length of a word that reaches $\MBZ_n \setminus X$ is at least $i + 1 \ge x+1$.
	\end{enumerate}
\end{proof}

\begin{lemma} \label{lem:SRW}
	The length of a word that reaches $\mathbb{Z}_n \setminus \{n-1,\frac{n}{2}-1\}$ in $\MCA$ is at least $\frac{5}{2}n-3$.
\end{lemma}
\begin{proof}
	Let $w$ be a shortest word that reaches $T$. Then $w$ can be decomposed into several words in $\Pi$, that is, $w = w_1 w_2 \cdots w_t$ for some integer $t$ and $w_1, \ldots, w_t \in \Pi$. 
	Write $S$ for $\mathbb{Z}_n \setminus \{n-1,\frac{n}{2}-1\}$. Let $T = \delta^{-1}(S, w_t)$ and $R = \delta^{-1}(T, w_{t-1})$
	By \cref{lem:predecessor}~\cref{itm:1}, we have $T = \mathbb{Z}_n \setminus \{n-1,n-2\}$ and then  $w_t = ab^i$ where $i \ge \frac{n}{2}-1$.
	By \cref{lem:predecessor}~\cref{itm:2}, we have $\min(R)\in \{n-3,n-4\}$. The rest of the proof is divided into two cases.

	\medskip
	\noindent{\textbf{Case 1.} $\min(R) = n-3$}

	Since $n-3 \notin R$, $0,1 \notin \delta(R,a)$. Then $w_{t-1} = ab^j$ such that $j \ge n-2$.
	Using \cref{lem:predecessor}~\cref{itm:3}, $|w_1 \cdots w_{t-2}| \ge n-3$. Hence $w \ge (1+n-3) + (1+ n-2) + (1 + \frac{n}{2}-1) = \frac{5}{2}n- 3$.

	\medskip
	\noindent{\textbf{Case 2.} $\min(R) = n-4$}

	Since $n-4 \notin R$, $0,n-1 \notin \delta(R,a)$. Then $w_{t-1} = ab^j$ such that $j \ge n-1$.
	Using \cref{lem:predecessor}~\cref{itm:3}, $|w_1 \cdots w_{t-2}| \ge n-4$. Hence $w \ge (1+n-4) + (1+ n-1) + (1 + \frac{n}{2}-1) = \frac{5}{2}n- 3$.

\end{proof}

\begin{proof}[Proof of \cref{thm:non-Don}]
	Let $n$ be any integer at least $10$.
	Since $\frac{5}{2}n- 3 > 2n$, by \cref{lem:CR,lem:SRW}, $\MCA_n$ is a binary completely reachable automaton and $\MCA_n$ does not fulfill Don's conjecture.
\end{proof}

Let $\MCA = (\MBZ_n, \{a,b\},\delta)$ be a binary completely reachable DFA. 
For a positive integer $k$, set $\MCA^{k}$ be the binary DFA $(\MBZ_n, \{a_k, b\}, \delta_k)$ such that $\delta_k(q, a_k) = \delta(q, b^ka)$ and $\delta_k(q,b) = \delta(q,b)$ for every $q \in \MBZ_n$. One can easily show that $\MCA^k$ is compeltely reachable.
Let $q_1$ and $q_2$ be the two distinct elements in $\dupl(\MCA)$. Both $\MCA^{q_1}$ and $\MCA^{q_2}$ are standardized completely reachable DFA and each of these two DFAs is called the \emph{standardization} of $\MCA$.

\begin{remark}
	The standardizations of $\MCA_n$ fulfill the Don's conjecture, for every even integer $n \ge 10$.
\end{remark} 
\begin{proof}
	The orbits of the standardizations of $\MCA_n$ are $\{n/2, n/2-2\}$ and $\{n/2, n/2 - 1\}$.
	Their orbit subgroups are $(2\MBZ_n,\oplus)$ and $(\MBZ_n,\oplus)$. By \cref{thm:Don}, the standardizations of $\MCA_n$ fulfill the Don's conjecture.
\end{proof}

\section{Conclusion} \label{sec:conclusion}

We have confirmed Don's conjecture for a subfamily of standardized completely reachable DFAs. Moreover, we prove that in an $n$-state standardized completely reachable DFA, every
$s$-element subset of states has a reaching word of length at most $n(n-s) + n - 1$. On the other hand, we contruct infinitely many binary completely reachable DFAs which do not fulfill Don's conjecture.

For a given DFA $\MCA$ and a subset $S$ of states, to determine whether they fulfill Don's conjecture can be seen as a ``qualitative analysis'' of the length of a shortest word that reaches $S$ in $\MCA$. More importantly and more difficult is its ``quantitative analysis''. We will do some brief discussions. 

For an $n$-state completely reachable automaton $\MCA = (Q, \Sigma, \delta)$ and a subset $S$ of its state set, let $f_{\MCA} (S)$ be the length of a shortest word that reaches $S$.
Let $\mathsf{C}$ be a subclass of completely reachable DFAs. For positive integers $n$ and $s$, define 
\[
	f_{\mathsf{C}}(n,s) := \max_{\MCA} \max_S \{ f_{\MCA}(S)\}
\]
where $\MCA$ runs over all $n$-state DFA in $\mathsf{C}$ and $S$ runs over all $s$-element subset of states in $\MCA$.
The restriction of Don's conjecture to $\mathsf{C}$ can be stated as $f_{\mathsf{C}}(n,s) \le n(n-s)$ for all $0 < s \le n$.

Write $\mathsf{BCR}$ for the class of binary completely reachable DFAs and $\mathsf{Std}$ for the class of standardized completely reachable DFAs.
\cref{thm:upperBound} shows that 
\[
	f_{\mathsf{Std}}(n,s) \le n(n-s) + n -1
\]
for all $0 < s \le n$. 
In \cref{sec:doNotfulfill}, by constructing $\MCA_n$ in \cref{ex:counterexample}, we obtain that 
\[
	f_{\mathsf{BCR}}(n,n-2) \ge \frac{5}{2}n - 3
\] 
for $n \ge 10$. Unfortunately, our construction cannot give any non-trivial lower bound for  $f_{\mathsf{BCR}}(n,s)$ when $s \neq n-2$.

Finally, we present a relation between the reaching words in a binary completely reachable DFA and its standradization. 
Let $\MCA = (\MBZ_n, \{a,b\}, \delta)$ be a binary completely reachable DFA and let $\MCA^k$ be a standardization of $\MCA$. 
Let $S$ be an $s$-element non-empty proper subset of $\MBZ_n$.
In the proof of \cref{thm:upperBound}, we construct a word $w \in \{a_k, b\}^*$ of length $\le (n-s)n + n-1$ such that $w$ reaches $S$ in $\MCA^k$. 
Further, $w$ contains at most $2n-s-1$ occurrences of $a_k$.
Let $w'$ be the word over $\{a,b\}$ which is obtained by replacing $a_k$ in $w$ with $b^ka$. Since $\delta_k(\cdot, a_k) = \delta(\cdot, b^ka)$, it holds that $\delta_k(Q, w) = \delta(Q, w') = S$. And then we have 
\begin{equation} \label{eq:standard}
	f_{\MCA} (S) \le (n-s)n + n -1 + k(2n-s -1).
\end{equation}
Recall that Ferens-Szyku{\l}a  bound \cite[Theorem 30]{FS23} shows that
\begin{equation} \label{eq:FS}
	f_{\MCA} (S) \le 2(n - s)n - n \ln(n - s) - n/(n- s).
\end{equation}
Observe that for any positive real $\epsilon$, in the case that $k < (\frac{1}{2} - \epsilon)n$,  \cref{eq:standard} provides an asymptotically  better bound of $f_{\MCA} (S)$ than \cref{eq:FS}. 

\section*{Acknowledgement}

I thank Prof. Mikhail V. Volkov for valuable discussions, feedback and research suggestions.

\end{document}